\documentclass[11pt]{article}

\usepackage[margin=1in]{geometry}

\usepackage{hyperref}
\usepackage{appendix}

\usepackage{graphicx,amssymb,amstext,amsmath,amsthm,amsfonts,verbatim,latexsym,mathrsfs}
\usepackage[sans]{dsfont}
\usepackage{fancyhdr}
\usepackage{ifthen}
\usepackage[noend]{algorithmic}
\usepackage{algorithm}

\newtheorem{lemma}{Lemma}[section]

\newtheorem{theorem}[lemma]{Theorem}

\newtheorem{claim}{Claim}

\newtheorem{corollary}[lemma]{Corollary}

\newtheorem{conjecture}[lemma]{Conjecture}

\theoremstyle{definition}

\newtheorem{definition}[lemma]{Definition}

\newcommand{\calC}{\mathcal{C}}

\newcommand{\bbC}{\mathbb{C}}

\newcommand{\bbF}{\mathbb{F}}

\newcommand{\bbR}{\mathbb{R}}

\newcommand{\bfR}{\mathbf{R}}

\newcommand{\rank}{\operatorname{rank}}

\newcommand{\mon}{\operatorname{mon}}

\newcommand{\poly}{\operatorname{poly}}

\newcommand{\eps}{\epsilon}

\newcommand{\fc}[1]{\widehat{#1}}

\newcommand{\comf}[2]{
\ifthenelse{\equal{#1}{1}}{
                     \ifthenelse{\equal{#2}{1}}{f:\{0,1\}^n \times \{0,1\}^n \to \{0,1\}}{}
}{}
\ifthenelse{\equal{#1}{-1}}{
                     \ifthenelse{\equal{#2}{1}}{f:\{-1,1\}^n \times \{-1,1\}^n \to \{0,1\}}{}
}{}
\ifthenelse{\equal{#1}{1}}{
                     \ifthenelse{\equal{#2}{-1}}{f:\{0,1\}^n \times \{0,1\}^n \to \{-1,1\}}{}
}{}
\ifthenelse{\equal{#1}{-1}}{
                     \ifthenelse{\equal{#2}{-1}}{f:\{-1,1\}^n \times \{-1,1\}^n \to \{-1,1\}}{}
}{}
}

\newcommand{\f}[1]{\textnormal{\textsc{#1}}}

\newcommand{\norm}[1]{ \| #1 \|}

\newcommand{\ip}[1]{\langle #1 \rangle}

\newcommand{\defeq}{\overset{\textnormal{def}}{=}}

\newcommand{\extra}[1]{}

\renewcommand{\comment}[1]{}
\newcommand{\nc}{\newcommand}
\nc{\ox}{\otimes}
\nc{\dg}{\dagger}
\nc{\dn}{\downarrow}
\nc{\cA}{{\cal A}}
\nc{\cB}{{\cal B}}
\nc{\cC}{{\cal C}}
\nc{\cD}{{\cal D}}
\nc{\cE}{{\mathcal E}}
\nc{\cF}{{\cal F}}
\nc{\cG}{{\cal G}}
\nc{\cH}{{\cal H}}
\nc{\cI}{{\cal I}}
\nc{\cJ}{{\cal J}}
\nc{\cK}{{\cal K}}
\nc{\cL}{{\cal L}}
\nc{\cM}{{\cal M}}
\nc{\cN}{{\cal N}}
\nc{\cO}{{\cal O}}
\nc{\cP}{{\cal P}}
\nc{\cR}{{\cal R}}
\nc{\cS}{{\cal S}}
\nc{\cT}{{\cal T}}
\nc{\cU}{{\cal U}}
\nc{\cX}{{\cal X}}
\nc{\cZ}{{\cal Z}}

\nc{\entI}{{\bf I}}
\nc{\entIarrow}{{\bf I}^{\leftarrow}}
\nc{\entH}{{\bf H}}
\nc{\entS}{{\bf S}}
\nc{\entHmin}{\mathbf{H}_{\min}}

\nc{\entF}{{\bf E}_f}

\nc{\isom}{\simeq}

\nc{\rar}{\rightarrow}
\nc{\lrar}{\longrightarrow}
\nc{\1}{{\openone}}

\nc{\weight}{\textbf{w}}
\nc{\hamdist}{d_{H}}

\nc{\Sp}{{{\mathbb S}}}
\nc{\RR}{{{\mathbb R}}}
\nc{\CC}{{{\mathbb C}}}
\nc{\FF}{{{\mathbb F}}}
\nc{\NN}{{{\mathbb N}}}
\nc{\ZZ}{{{\mathbb Z}}}
\nc{\PP}{{{\mathbb P}}}
\nc{\QQ}{{{\mathbb Q}}}
\nc{\UU}{{{\mathbb U}}}
\nc{\OO}{{{\mathbb O}}}
\nc{\EE}{{{\mathbb E}}}
\nc{\id}{{\operatorname{id}}}

\nc{\qubitchannel}{\id_2}
\nc{\bitchannel}{\overline{\id}_2}

\nc{\be}{\begin{equation}}
\nc{\ee}{{\end{equation}}}
\nc{\bea}{\begin{eqnarray}}
\nc{\eea}{\end{eqnarray}}
\nc{\<}{\langle}
\nc{\Hom}[2]{\mbox{Hom}(\CC^{#1},\CC^{#2})}
\nc{\rU}{\mbox{U}}

\nc{\ob}[1]{#1}

\newcommand{\ex}[1]	{\mathbf{E}\left[ #1 \right]}
\newcommand{\exd}[2]	{\mathbf{E}_{#1}\left[ #2 \right]}
\newcommand{\pr}[1]	{\mathbf{P}\left[ #1 \right]}

\newcommand{\eqdef}	{\stackrel{\textrm{def}}{=}}

\renewcommand{\exp}[1]	{\operatorname{exp}\left( #1 \right)}

\newcommand{\floor}[1]	{\left\lfloor #1 \right\rfloor}

\nc{\re}{\mathrm{Re}}

\newcommand{\xor}{\oplus}

\nc{\binent}{h}

\nc{\rc}{r^{\angle}}

\begin{document}

\title{Spectral Norm of Symmetric Functions}

\author{Anil Ada\footnote{School of Computer Science, McGill University. {\small {\tt aada@cs.mcgill.ca}}} \and Omar Fawzi\footnote{School of Computer Science, McGill University. {\small {\tt ofawzi@cs.mcgill.ca}}} \and Hamed Hatami\footnote{School of Computer Science, McGill University. {\small {\tt hatami@cs.mcgill.ca}}}}

\date{\today}

\maketitle

\begin{abstract}
The spectral norm of a Boolean function $f:\{0,1\}^n \to \{-1,1\}$ is the sum of
the absolute values of its Fourier coefficients. This quantity provides useful upper
and lower bounds on the complexity of a function in areas such as learning theory,
circuit complexity, and communication complexity. In this paper, we give a combinatorial
characterization for the spectral norm of symmetric functions. We show that
the logarithm of the spectral norm is of the same order of magnitude as $r(f)\log(n/r(f))$ where $r(f) =
\max\{r_0,r_1\}$, and $r_0$ and $r_1$ are the smallest integers less than $n/2$
such that $f(x)$ or $f(x) \cdot \f{parity}(x)$ is constant for all $x$ with $\sum x_i \in [r_0, n-r_1]$.
We mention some applications to the decision tree and communication complexity of symmetric functions.
\end{abstract}
\section{Introduction}

The study of Boolean functions $f:\{0,1\}^n \to \{-1,1\}$ is central to complexity theory and combinatorics as objects of interest in these areas can often be represented as Boolean functions. Fourier analysis of Boolean functions provides some of the strongest tools in this study with applications to graph theory, circuit complexity, communication complexity, hardness of approximation, machine learning, etc.


In many different settings, Boolean functions with ``smeared out'' Fourier spectrums have higher ``complexity''. There are various useful ways to measure the spreadness of the spectrum. Some notable ones are the spectral norm $\norm{\fc{f}}_1 = \sum_S |\fc{f}(S)|$ (i.e., the $\ell_1$ norm), the $\ell_\infty$ norm $\norm{\fc{f}}_\infty = \max_S |\fc{f}(S)|$, and the Shannon entropy of the squares of the Fourier coefficients $H[\fc{f}^2] = - \sum_S \fc{f}(S)^2 \log \fc{f}(S)^2$. The focus of this paper is on the spectral norm.

\subsubsection*{Spectral Norm of Boolean Functions}\label{sec:history}

As $\sum_S \fc{f}(S)^2 = 1$ for a Boolean function $f$, it is often useful to view the squares of the Fourier coefficients as a probability distribution over the subsets $S \subseteq [n]$. The spectral norm corresponds to the R\'enyi entropy of order $1/2$ of the squares of the Fourier coefficients, $H_{1/2}[\fc{f}^2] = 2 \log \left ( \sum_S |\fc{f}(S)| \right) = 2 \log \norm{\fc{f}}_1$. It provides useful upper and lower bounds on the \emph{complexity} of a function in settings such as learning theory, circuit complexity, and communication complexity. It is particularly useful in the settings where $\f{parity}$ is considered a function of low complexity. We list some of the applications below.

In the setting of learning theory, the spectral norm is used in conjunction with the \emph{Kushilevitz-Mansour Algorithm} \cite{KM91}. This algorithm, using membership queries, learns efficiently a concept class $\mathcal{C}$ where the Fourier spectrum of every function in $\mathcal{C}$ is concentrated on a small set of characters (This set can be different for different functions.). Kushilevitz and Mansour observe that an upper bound on the spectral norm implies such a concentration, and obtain:
\begin{quote}
 If $\mathcal{C} = \{f:\{0,1\}^n \to \{-1,1\} \; | \; \norm{\fc{f}}_1 \leq s\}$, then $\calC$ is learnable with membership queries in time $\poly(n,s,1/\eps)$.
\end{quote}
Using the above result, they show that functions computable by small size parity decision trees\footnote{Parity decision trees generalize the usual decision tree model: in every node we branch according to the parity of a subset of the variables.} are efficiently learnable with membership queries. This is done by observing that a function computable by a size $s$ parity decision tree satisfies $\norm{\fc{f}}_1 \leq s$. This inequality is also interesting since it provides a lower bound  in terms of the spectral norm on the size of any parity decision tree computing $f$.


Threshold circuits (i.e., circuits composed of threshold gates) constitute an important model of computation (in part due to their resemblance to neural networks), and they have been studied extensively. A classical result of Bruck and Smolensky \cite{BS92} states that a function with small spectral norm can be represented  as the sign of a polynomial with few monomials. This in turn implies that functions with small spectral norm can be computed by depth 2 threshold circuits of small size. The result of Bruck and Smolensky has found other interesting applications (see for example \cite{SB91,GHR92,Gro99,OS08}).

We now turn our attention to communication complexity. Arguably the most famous conjecture in communication complexity is the Log Rank Conjecture which states that the deterministic communication complexity of a function $F: \{0,1\}^n \times \{0,1\}^n \to \{-1,1\}$ is upper bounded by $\log^c \rank M_F$ where the matrix $M_F$ is defined as $M_F[x,y] = F(x,y)$. Grolmusz \cite{Gro97} makes a similar intriguing conjecture for the randomized communication complexity:
\begin{quote}
 There is a constant $c$ such that the public coin randomized communication complexity of $F: \{0,1\}^n \times \{0,1\}^n \to \{-1,1\}$ is upper bounded by $\log^c \norm{\fc{F}}_1$.
\end{quote}
In the same paper, Grolmusz is able to prove a much weaker upper bound of $O(\norm{\fc{F}}_1^2 \delta(n))$ with $\exp{-c\delta(n)}$ probability of error. Even this weaker result has interesting applications in circuit complexity and decision tree complexity (see \cite{Gro97} for more details).

Another major open problem in communication complexity is whether the classical and quantum communication complexity of total Boolean functions $f: X \times Y \to \{-1, 1\}$ (i.e., functions defined on all of $X \times Y$) are polynomially related. It is conjectured that this is so and research has been focused on establishing it for natural large families of functions. In an important paper \cite{Raz03}  Razborov showed that the conjecture is true for functions of the form $F(x,y) = \f{sym}(x \wedge y)$  where $\f{sym}$ denotes a symmetric function, and $x \wedge y$ is the bitwise $\f{and}$ of $x$ and $y$. Shi and Zhang \cite{SZ09a} verified the conjecture for $F(x,y)= \f{sym}(x \xor y)$  where $x \xor y$ denotes the bitwise $\f{xor}$. The next big targets are $F(x,y) = f(x \wedge y)$ and $F(x,y) = f(x \xor y)$ for general $f$, but handling arbitrary $f$ seems difficult at the moment.

A variant of the spectral norm, \emph{the approximate spectral norm}, is intimately related to the communication complexity of ``xor functions''. The $\eps$-approximate spectral norm of $f$, denoted $\norm{\fc{f}}_{1,\eps}$, is the smallest spectral norm of a function $g:\{0,1\}^n \to \bbR$ such that $\|f - g\|_\infty \leq \eps$. It is known (see for example \cite{LS09}) that $\log \norm{\fc{f}}_{1,\eps}$ lower bounds the quantum bounded error communication complexity of $f(x \xor y)$. We expect that the lower bound $\log \norm{\fc{f}}_{1,\eps}$ is tight, and that this quantity characterizes the communication complexity of xor functions. More discussion on the communication complexity of xor functions, and how it relates to this work is given in Section \ref{sec:futurework}.

This ends our discussion of the use of the spectral norm in learning theory, circuit complexity and communication complexity. We conclude this subsection by mentioning a relatively recent result that studies the spectral norm of Boolean functions. Green and Sanders \cite{GS08} show that every Boolean function whose spectral norm is bounded by a constant can be written as a sum of constantly many $\pm$ indicators of cosets. This gives an interesting characterization of Boolean functions with small spectral norm.

\subsubsection*{Fourier Spectrum of Symmetric Functions}

A function $f : \{0,1\}^n \to \{-1,1\}$ is called \emph{symmetric} if it is invariant under permutations of the coordinates. In other words the value of $f(x)$  depends only on $\sum x_i$ (i.e., $f(x) = f(y)$ whenever $\sum_i x_i = \sum_i y_i$). Symmetric functions are at the heart of complexity theory as natural functions like $\f{and}$, $\f{or}$, $\f{majority}$, and $\f{mod}_m$ are all symmetric. They are often the starting point of investigation because the symmetry of the function can be exploited. On the other hand, they can also have surprising power. In several settings, functions such as $\f{parity}$ and $\f{majority}$ represent ``hard'' functions. Given their central role, it is of interest to gain insight into the Fourier spectrum of symmetric functions.

There are various nice results related to the Fourier spectrum of symmetric functions. We cite a few of them here. A beautiful result of Paturi \cite{Pat92} tightly characterizes the approximate degree of every symmetric function, and this has found many applications in theoretical computer science \cite{Raz03,BBCMW01,She09,Wol08, She11}. Kolountzakis {\it et al.} \cite{KLMMV09} studied the so called \emph{minimal degree} of symmetric functions and applied their result in learning theory. Shpilka and Tal \cite{ST11} later simplified and improved the work of Kolountzakis {\it et al.} Recently, O'Donnell, Wright and Zhou \cite{OWZ11} verified an important conjecture in the analysis of Boolean functions, the Fourier Entropy/Influence Conjecture, in the setting of symmetric functions. In fact we make use of their key lemma in this paper.


\subsection{Our Results and Proof Overview}


We give a combinatorial characterization of the spectral norm of symmetric functions. For $x \in \{0,1\}^n$, define $|x| \eqdef \sum x_i$. For a
function $f:\{0,1\}^n \to \{-1,1\}$, let $r_0$ and $r_1$ be the minimum integers less than $n/2$ such that $f(x)$ or $f(x) \cdot \f{parity}(x)$  is constant for $x$ with $|x| \in [r_0, n-r_1]$. Define $r(f) \eqdef \max \{r_0,r_1\}$. We show that $\log \norm{\fc{f}}_1$ is of the same order of magnitude as $r(f)\log(n/r(f))$:

\begin{theorem}[Main Theorem]\label{thm:l1-sym}
For any symmetric function $f : \{0,1\}^n \to \{-1, 1\}$, we have
\[
\log \| \fc{f} \|_1 = \Theta \left( r(f) \log\left(\frac{n}{r(f)}\right) \right)
\]
whenever $r(f) >  1$. If $r(f) \le 1$, then $\| \fc{f} \|_1 = \Theta(1)$.
\end{theorem}

As an application, we give a characterization of the parity decision tree size of symmetric functions. As mentioned in Section~\ref{sec:history}, a parity decision tree computes a boolean function by querying the parities of subsets of the variables. The size of the tree is simply the number of leaves in the tree.
\begin{corollary}
\label{cor:par-dec-tree}
Let $f: \{0,1\}^n \to \{-1,1\}$ be a symmetric function. Then the parity decision tree size of $f$ is $2^{\Theta(r(f) \log(n/r(f)))}$.
\end{corollary}
The proof of this corollary is presented in Section~\ref{sec:cor}. Note that the lower bound also applies in the case of the usual decision tree size (where one is restricted to query only variables). Decision tree size is an important measure in learning theory; algorithms for learning decision trees efficiently is of great interest both for practical and theoretical reasons. One of the most well-known and studied problems is whether small size decision trees are efficiently learnable from uniformly random examples.

As a second application, using the protocol of Shi and Zhang \cite[Proposition 3.4]{SZ09a}, and the observation that $\norm{\fc{F}}_1 = \norm{\fc{f}}_1$ when $F(x,y) = f(x \xor y)$, we verify Grolmusz's conjecture mentioned earlier in Section~\ref{sec:history} in the setting of \emph{symmetric xor functions}.

\begin{corollary}
 Let $f: \{0,1\}^n \to \{-1,1\}$ be a symmetric function and let $F :\{0,1\}^n \times \{0,1\}^n \to \{-1,1\}$ be defined as $F(x,y) = f(x \xor y)$. Then the public coin constant error randomized communication complexity of $F$ is upper bounded by $O(\log^2 \norm{\fc{F}}_1)$.
\end{corollary}

We now give an outline for the proof of Theorem \ref{thm:l1-sym}. The upper bound is quite straightforward and is given in Lemma \ref{lem:upper-bound}. The lower bound is handled in two different cases: when $r(f)$ is bounded away from $n/2$ (Lemma \ref{lem:lb-small-r}) and when $r(f)$ is close to $n/2$ (Lemma \ref{lem:lb-large-r}).

We refer to the Fourier spectrum of $f$ restricted to the sets $S \subseteq [n]$ of size $k$ as the \emph{$k$-th level} of the Fourier spectrum. Note that for a symmetric $f$, we have $\fc{f}(S) = \fc{f}(T)$ whenever $|S| = |T|$. Therefore the Fourier spectrum is maximally spread out in each level. The overall strategy for the lower bound is to show an appropriate lower bound on the $\ell_2$ mass of the Fourier spectrum on a middle level. Middle levels have many Fourier coefficients, and therefore contribute significantly to the spectral norm provided there is enough $\ell_2$ mass on them. An important tool in our analysis is the use of certain discrete derivatives of $f$. Identify $\{0,1\}^n$ with $\mathbb{F}_2^n$ and let $e_1,\ldots,e_n$ 
denote the standard vectors in $\mathbb{F}_2^n$. For $i \neq j$, define $f_{ij}(x) \defeq f(x+e_i+e_j)-f(x)$. We observe that
\[
 \sum_{i \neq j} \ex{f_{ij}^2} = 8\sum_{S} |S| (n-|S|) \fc{f}(S)^2.
\]
The quantity on the LHS, and therefore the RHS, can be lower bounded using $r(f)$ (Lemma \ref{lem:bias-r}). As the coefficient $|S|(n-|S|)$ increases as $|S|$ approaches $n/2$, we are able to give a lower bound on the $\ell_2$ mass of the Fourier spectrum on the middle levels. This approach gives tight bounds for $r(f)$ bounded away from $n/2$, but not for  a function such as $\f{majority}$.

To handle functions $f$ with $r(f)$ close to $n/2$, we use ideas from \cite{OWZ11}. The main lemma of \cite{OWZ11} states that the first derivatives of a symmetric function are noise sensitive. We observe that this is also true for the  derivatives $f_{ij}$. This allows us to derive the inequality
\[
  \sum_S |S|(n-|S|)\fc{f}(S)^2 (\rho^{|S|} + \rho^{n-|S|}) \leq  \frac{8}{\sqrt{\pi c}} \cdot \sum_S |S|(n-|S|)\fc{f}(S)^2,
\]
where $\rho = (1-c/n)$. The quantity $\rho^{|S|} + \rho^{n-|S|}$ is decreasing in $|S|$ for $|S| \leq n/2$. Thinking of $c$ as a large constant, we see that the dampening of the middle levels with $\rho^{|S|} + \rho^{n-|S|}$ decreases the value of the sum significantly. From this, we can lower bound the $\ell_2$ mass of the middle levels. Note that if $\sum_S |S|(n-|S|)\fc{f}(S)^2$ is small to begin with ($r(f)$ is small), the above inequality is not useful. On the other hand if $r(f)$ is large, $\sum_S |S|(n-|S|)\fc{f}(S)^2$ is large, and the strategy just described gives good bounds.

\section{Preliminaries}

We view Boolean functions $f:\{0,1\}^n \to \{-1,1\}$ as residing in the vector space $\{f: \{0,1\}^n \to \bbC\}$. If we view the domain as the group $\bbF_2^n$, we can appeal to Fourier analysis, and express every $f:\{0,1\}^n \to \bbC$ (uniquely) as a linear combination of the characters of $\mathbb{F}_2^n$. That is every function $f: \bbF_2^n \to \bbC$ can be written as
$
 f= \sum_{S \subseteq [n]} \fc{f}(S) \chi_S,
$
where the characters $\chi_S$ are defined as
$
 \chi_S: x \mapsto  (-1)^{\sum_{i \in S} x_i},
$
and $\fc{f}(S) \in \bbC$ are their corresponding Fourier coefficients. Since the characters form an orthonormal basis for $\{f: \{0,1\}^n \to \bbC\}$, we have
$
 \fc{f}(S) = \ip{f,\chi_S} = \exd{x}{ f(x)\chi_S(x)}.
$

For a Boolean function $f$, we define $W_k[f] = \sum_{|S| = k} |\fc{f}(S)|^2$. We simply use $W_k$ when $f$ is clear from the context. For a symmetric function, we often write $f(k)$ for $f(x)$ with $\sum_i x_i = k$ and $k \in [n]$. We use $\binent$ to denote the binary entropy function $\binent(\alpha) = -\alpha \log(\alpha) - (1-\alpha) \log(1-\alpha)$. We will use the following simple estimates for binomial coefficients (See~\cite[Lemma 9.2]{MU05}): 
Let $\alpha \in [0,1]$ such that $\alpha n$ is an integer. Then
\begin{equation}
\label{eq:B2-1}
\sum_{k=0}^{\alpha n} \binom{n}{k} \leq 2^{n \binent(\alpha)},
\end{equation}
and
\begin{equation}
\label{eq:B2-2}
\frac{2^{n \binent(\alpha)}}{n+1} \leq \binom{n}{\alpha n}.
\end{equation}
If $\alpha \in [0,1/2]$ is arbitrary, then
\begin{equation}
\label{eq:B2-3}
\frac{2^{n \binent(\alpha)}}{n(n+1)} \leq \binom{n}{\floor{\alpha n}} \leq 2^{n \binent(\alpha)}.
\end{equation}

The following fact is also easy and classical. For every constant $c > 0$, there exists a constant $C > 0$ such that for any $n \geq 1$,
\begin{equation}
\label{eq:B3}
\binom{n}{\floor{n/2+c \sqrt{n}}} \geq C \frac{2^n}{\sqrt{n}}.
\end{equation}

\begin{definition}
\label{def:r}
For any $f : \{0,1\}^n \to \RR$, we define
\[
R(f) \eqdef \sum_{S \subseteq [n]} |S| (n-|S|) \fc{f}(S)^2.
\]
\end{definition}

For $a \in \mathbb{F}_2^n$, we define the derivative of $f:\mathbb{F}_2^n \to \mathbb{R}$ in the direction $a$ as
$$\Delta_af:x \mapsto f(x+a)-f(x).$$ 
Let $e_1,\ldots,e_n$ denote the standard vectors in $\mathbb{F}_2^n$, and let $f : \{0,1\}^n \to \RR$. For all $i \neq j$, define
\begin{equation}
\label{eq:second-deriv}
f_{ij} \defeq \Delta_{e_i+e_j}f.
\end{equation}

\begin{lemma} \label{lem:r}
For every $f:\{0,1\}^n \to \mathbb{R}$, we have
\[
\sum_{i \neq j} \ex{f_{ij}^2} = 8 R(f). 
\]
\end{lemma}
\begin{proof}
We have $$f_{ij}(x)= \sum_S \fc{f}(S)\chi_S(x) (\chi_S(e_i+e_j)-1)=\sum_{S: |S \cap \{i,j\}|=1} -2\fc{f}(S)\chi_S(x),$$
which by Parseval's identity implies 
$$\ex{f_{ij}^2}=\sum_{S: |S \cap \{i,j\}|=1} 4 \fc{f}(S)^2.$$
Summing over all pairs $i \neq j$, we obtain
\[
\sum_{i \neq j} \ex{f_{ij}^2} = 8\sum_{S \subseteq [n]} |S| (n-|S|) \fc{f}(S)^2.
\]
\end{proof}

\section{Proof of Theorem~\ref{thm:l1-sym}}
%
%
%
As mentioned earlier the upper bound is proved in Lemma \ref{lem:upper-bound}. The proof of the lower bound is divided into two parts: Lemma \ref{lem:lb-small-r} handles the case where $r$ is bounded away from $n/2$ and Lemma \ref{lem:lb-large-r} the case when $r$ is close to $n/2$.

\subsection{Upper Bound} \label{sec:upper-bound}

\begin{lemma} \label{lem:upper-bound}
For all $n \geq 1$ and every symmetric function $f : \{0,1\}^n \to \{-1,1\}$,
\[
\log \| \fc{f} \|_1 \leq 2 \cdot  r(f) \log(n/r(f)) + 3.
\]
\end{lemma}
\begin{proof}
By definition of $r_0$ and $r_1$,  there exists a function $p \in \{-1, 1, -\f{parity}, +\f{parity}\}$ such that $f(k) = p(k)$ for all $k \in [r_0, r_1]$. By linearity of the Fourier transform, we have for any $S \subseteq [n]$,
\begin{align*}
\fc{f}(S) &= \fc{p}(S) + \fc{f-p}(S) \\
		&= \fc{p}(S) + \frac{1}{2^n} \sum_{k=0}^{n} (f(k) - p(k)) \sum_{|x| = k} \chi_S(x) \\
		&= \fc{p}(S) + \frac{1}{2^n} \sum_{k=0}^{r_0-1} (f(k) - p(k)) \sum_{|x| = k} \chi_S(x) + \frac{1}{2^n} \sum_{k=n-r_1+1}^{n} (f(k) - p(k)) \sum_{|x| = k} \chi_S(x)		
\end{align*}
Thus,
\begin{align*}
|\fc{f}(S)|&\leq |\fc{p}(S)| + \frac{1}{2^n} \sum_{k=0}^{r_0-1} 2 \binom{n}{k} + \frac{1}{2^n} \sum_{k=n-r_1+1}^{n} 2 \binom{n}{k} \\
		&\leq |\fc{p}(S)| + 2 \cdot \frac{2^{h(r_0/n) n} + 2^{h(r_1/n) n} }{2^n}. 		
\end{align*}
For the last inequality, we used \eqref{eq:B2-1}. Summing over all subsets $S \subseteq [n]$, we get
\[
\| \fc{f} \|_1 \leq 1 + 2 (2^{\binent(r_0/n) n} + 2^{\binent(r_1/n) n}) \leq 1 + 4 \cdot 2^{h(r/n) n}.
\]
As $\binent(t) \leq - 2 t \log t$ when $t \leq 1/2$, we obtain
$
\log \| \fc{f} \|_1 \leq 3 + 2 r \log(n/r).
$
\end{proof}


\subsection{Lower Bound}

We start by making some simple observations.
\begin{lemma}\label{lem:bias-r}
Let $f:\{0,1\}^n \to \{-1,1\}$ be a symmetric function, and define $r_0 = r_0(f)$ and $r_1 = r_1(f)$. Then
\begin{equation}
\label{eq:lower-bound-r}
R(f) \geq \left( (n-r_0 + 1) ( n-r_0) \binom{n}{r_0-1} + (n - r_1 + 1) (n-r_1) \binom{n}{r_1 - 1} \right) 2^{-n}.
\end{equation}
Moreover, assuming that $f(s) = 1$ for all $s \in \{r_0, \dots, n-r_1\}$, we have
\begin{equation}
\label{eq:lower-bound-bias}
\sum_{S \neq \emptyset} \fc{f}(S)^2 \leq 4  \left(\sum_{s < r_0} \binom{n}{s} + \sum_{s < r_1} \binom{n}{s} \right) 2^{-n}.
\end{equation}
\end{lemma}
\begin{proof}
Define $f_{ij}$ as in \eqref{eq:second-deriv}. As $f$ is symmetric, we only need to consider $f_{12}$.
\begin{align*}
\ex{f_{12}^2} 
&= \exd{x_3\dots x_n}{ \frac{1}{4} \cdot \left( f_{12}^2(00 x_3 \dots x_n ) + f_{12}^2(01 x_3 \dots x_n ) + f_{12}^2(10 x_3 \dots x_n ) +  f_{12}^2(11 x_3 \dots x_n) \right) } \\
&= \frac{1}{4} \exd{x_3\dots x_n}{ \left(f(00 x_3 \dots x_n ) - f(11x_3 \dots x_n) \right)^2 + \left(f(11 x_3 \dots x_n) - f(00x_3 \dots x_n) \right)^2 } \\
&\geq \frac{1}{2} \left( \binom{n-2}{r_0-1} \cdot 2^{-(n-2)} \cdot 4 + \binom{n-2}{n- r_1 - 1} \cdot 2^{-(n-2)} \cdot 4 \right) \\
&= 8 \cdot \left(\frac{(n-r_0+1)(n-r_0)}{n(n-1)} \cdot \binom{n}{r_0-1} + \frac{(n-r_1+1)(n-r_1)}{n(n-1)} \cdot \binom{n}{r_1-1} \right) 2^{-n}.
\end{align*}
Inequality \eqref{eq:lower-bound-r} follows by applying Lemma \ref{lem:r}.

In order to establish inequality \eqref{eq:lower-bound-bias}, we show a lower bound on the principal Fourier coefficient of $f$:
\[
\fc{f}(\emptyset) \geq 1 - 2 \left(\sum_{s < r_0} \binom{n}{s} + \sum_{s > n-r_1} \binom{n}{s} \right) 2^{-n},
\]
which implies that
\[
\fc{f}(\emptyset)^2 \geq 1 - 4 \cdot \left(\sum_{s < r_0} \binom{n}{s} + \sum_{s < r_1} \binom{n}{s} \right) 2^{-n}.
\]
\end{proof}

\subsubsection{Lower Bound: $r \ll n/2$}
\begin{lemma}
\label{lem:lb-small-r}
For every symmetric function $f : \{0,1\}^n \to \{-1,1\}$ with $r = r(f)$,
\[
\log \| \fc{f} \|_1 \geq \Omega \left( \left(1 - \frac{2r-2}{n} \right) \cdot r \log(n/r)  \right).
\]
\end{lemma}
\begin{proof}
Observe that we can assume without loss of generality that $f(s) = 1$ for all $s \in \{r_0, \dots, n-r_1\}$. In fact, to handle the case $f = -1$ or $f = \pm \f{parity}$ in $[r_0,n-r_1]$, it suffices to multiply the function by $-1$ or by $\pm \f{parity}$, respectively. This does not affect the spectral norm of the function.

We prove the statement by showing that a significant portion of the $\ell_2$ mass of $\fc{f}$ sits in the middle levels from $m$ to $n-m$ for a well-chosen $m$ depending on $r(f)$.

Define $\alpha_0 = \frac{r_0 - 1}{n} < 1/2$ and $\alpha_1 = \frac{r_1 - 1}{n}$. We also let $m_0 = \floor{n/2 \cdot (1 - \sqrt{4 \alpha_0 - 6 \alpha_0^2 + 4 \alpha_0^3})}$ and $m_1 = \floor{n/2 \cdot (1 - \sqrt{4 \alpha_1 - 6 \alpha_1^2 + 4 \alpha_1^3})}$.

By Lemma \ref{lem:bias-r}, we have $\sum_{k > 0} W_k \leq 4 \cdot \left(\sum_{s < r_0} \binom{n}{s} + \sum_{s < r_1} \binom{n}{s} \right) 2^{-n}$. Let $U_k$ and $V_k$ be so that $W_k = U_k + V_k$ and $\sum_{k > 0} U_k \leq 4 \cdot 2^{-n} \sum_{s < r_0} \binom{n}{s}$ and $ \sum_{k > 0} V_k \leq 4 \cdot 2^{-n} \sum_{s < r_1} \binom{n}{s} 2^{-n}$.
Our objective is now to obtain a lower bound on $\sum_{k=m_0}^{n-m_0} k (n-k) U_k + \sum_{k=m_1}^{n-m_1} k (n-k) V_k$ using Lemma \ref{lem:bias-r}
\begin{align}
&\sum_{k=m_0}^{n-m_0}  k (n - k)  U_k + \sum_{k=m_1}^{n-m_1} k (n-k) V_k \; = \; R(f) - \sum_{k \notin [m_0, n-m_0]} k (n-k) U_k - \sum_{k \notin [m_1, n-m_1]} k (n-k) V_k \notag \\
&\quad \geq (n-r_0) (n-r_0+1) \binom{n}{r_0-1} 2^{-n} - (m_0-1)(n-m_0+1) 4 \cdot 2^{-n} \sum_{s < r_0} \binom{n}{s} \notag \\
&\qquad + (n-r_1) (n-r_1+1) \binom{n}{r_1-1} 2^{-n} - (m_1-1)(n-m_1+1) 4 \cdot 2^{-n} \sum_{s < r_1} \binom{n}{s}.
\label{eq:a0-a1}
\end{align}
Define $A_0 \eqdef (n-r_0) (n-r_0+1) \binom{n}{r_0-1} 2^{-n} - (m_0-1)(n-m_0+1) 4 \cdot 2^{-n} \sum_{s < r_0} \binom{n}{s}$, and let $A_1$ be its analogue
for $r_1$ so that the right hand side of \eqref{eq:a0-a1} equals $A_0+A_1$.

Observe that $\binom{n}{s} = \frac{s+1}{n-s} \binom{n}{s+1}$, and $\frac{s+1}{n-s} \leq \frac{r_0 - 1}{n-(r_0 - 1)} = \frac{\alpha_0}{1-\alpha_0}$ for $s < r_0-1$. Thus
\begin{align}
A_0 &\geq \binom{n}{r_0-1} 2^{-n} \left( (n - \alpha_0 n -1) (n - \alpha_0 n) - 4 (m_0-1) (n-m_0+1) \frac{1}{1-\alpha_0/(1-\alpha_0)} \right) \notag \\
&\geq \binom{n}{r_0-1} 2^{-n} \left( n^2 (1 - \alpha_0)^2 - (1-\alpha_0) n - 4 (m_0-1) (n-m_0+1) \frac{1-\alpha_0}{1-2\alpha_0} \right) \notag \\
&\geq \binom{n}{r_0-1} 2^{-n} \left( n^2 \left( (1 - \alpha_0)^2 - (1 - (4 \alpha_0 - 6 \alpha_0^2 + 4 \alpha_0^3)) \frac{1-\alpha_0}{1-2\alpha_0} \right) - (1-\alpha_0) n \right) \notag \\
&=\binom{n}{r_0-1} 2^{-n} (1-\alpha_0)  \left( n^2 \left( (1 - \alpha_0) - (1-2\alpha_0 + 2 \alpha_0^2) \right) - n \right) \notag \\
&=\binom{n}{r_0-1} 2^{-n} (1-\alpha_0)  \left( \alpha_0 (1-2\alpha_0) n^2  - n \right). \label{eq:a0}
\end{align}
Analogously, we have
\begin{align}
\label{eq:a1}
A_1 &\geq \binom{n}{r_1-1} 2^{-n} (1-\alpha_1)  \left( \alpha_1 (1-2\alpha_1) n^2  - n \right).
\end{align}
We now assume that $r_0 \geq r_1$. Observe that we then have $m_0 \leq m_1$. Combining \eqref{eq:a0-a1} and \eqref{eq:a0}, we get
\[
n^2 \sum_{k=m_0}^{n-m_0} W_k \geq \sum_{k=m_0}^{n-m_0} k (n-k) W_k \geq \binom{n}{r_0-1} 2^{-n} (1-\alpha_0)  \left( \alpha_0 (1-2\alpha_0) n^2  - n \right).
\]
Note that for symmetric functions $\| \fc{f} \|_1 = \sum_{k=0}^n \sqrt{\binom{n}{k} W_k}$, and thus
\begin{align}
\| \fc{f} \|_1 &\geq \sum_{k=m_0}^{n-m_0} \sqrt{\binom{n}{k} W_k}  \geq \sqrt{\binom{n}{m_0} \sum_{k=m_0}^{n-m_0} W_k} \notag \\
&\geq \sqrt{ \binom{n}{m_0} \binom{n}{r_0-1} 2^{-n} \frac{ (1-\alpha_0)  \left( \alpha_0 (1-2\alpha_0) n^2  - n \right) }{n^2} } \notag \\
&\geq \sqrt{ \binom{n}{\floor{n/2 (1 - \sqrt{4 \alpha_0 - 6 \alpha_0^2 + 4 \alpha_0^3})}
} \binom{n}{\alpha_0 n} 2^{-n} \frac{ (1-\alpha_0)  \left( \alpha_0 (1-2\alpha_0) n^2  - n \right) }{n^2} }. \label{eq:l1-lower-bound}
\end{align}
Using \eqref{eq:B2-2} and \eqref{eq:B2-3}, we obtain
\[
 \| \fc{f} \|^2_1 \geq \frac{2^{n \left(\binent\left(\frac{1}{2} - \frac{1}{2} \sqrt{4 \alpha_0 - 6 \alpha_0^2 + 4 \alpha_0^3}\right) + \binent(\alpha_0) - 1 \right) } }{ n(n+1)^2 } \cdot \frac{ (1-\alpha_0)  \left( \alpha_0 (1-2\alpha_0) n^2  - n \right) }{n^2}.
\]
As a result
\begin{align*}
\log \| \fc{f} \|_1 \geq \frac{n}{2} \left( \binent\left(\frac{1}{2} - \frac{1}{2}\sqrt{4 \alpha_0 - 6 \alpha_0^2 + 4 \alpha_0^3}\right) + \binent(\alpha_0) - 1 \right) + \frac{1}{2} \log \frac{ (1-\alpha_0)  \left( \alpha_0 (1-2\alpha_0) n^2  - n \right) }{n^3 (n+1)^2 }.
\end{align*}

\begin{claim}
There exists a constant $c>0$ such that for every $\alpha_0 \in (0,1/2)$,
\begin{equation}
\label{eq:claim}
\binent\left(\frac{1}{2} - \frac{1}{2}\sqrt{4 \alpha_0 - 6 \alpha_0^2 + 4 \alpha_0^3}\right) + \binent(\alpha_0) - 1 \geq c (1 - 2\alpha_0) \cdot \alpha_0 \cdot \log(1/\alpha_0).
\end{equation}
\end{claim}
\begin{proof}
Using the inequality $|h(x_2)-h(x_1)| \le h(x_2-x_1)$ which holds for every $0<x_1<x_2<1$, we have 
$$\binent\left(\frac{1}{2} - \frac{1}{2}\sqrt{4 \alpha_0 - 6 \alpha_0^2 + 4 \alpha_0^3}\right) + \binent(\alpha_0) - 1 
\ge \binent(\alpha_0)- \binent\left(\frac{1}{2}\sqrt{4 \alpha_0 - 6 \alpha_0^2 + 4 \alpha_0^3} \right).$$
By looking at the Taylor expansion, it is easy to see that there exists an $\epsilon>0$, such that for every $\alpha_0 \in [0,\epsilon] \cup \left[\frac{1}{2}-\epsilon,\frac{1}{2}\right]$ we have
$$ \binent(\alpha_0)- \binent\left(\frac{1}{2}\sqrt{4 \alpha_0 - 6 \alpha_0^2 + 4 \alpha_0^3} \right) \ge \frac{1}{2}(1 - 2\alpha_0) \cdot \alpha_0 \cdot \log(1/\alpha_0).$$
On the other hand, there exists a  constant $c_\epsilon>0$ such that  when $\alpha_0 \in (\epsilon,1/2-\epsilon)$, both $\binent(\alpha_0)- \binent\left(\frac{1}{2}\sqrt{4 \alpha_0 - 6 \alpha_0^2 + 4 \alpha_0^3} \right)$ and the right-hand side of (\ref{eq:claim}) belong to $[c_\epsilon,1]$. Taking $c \eqdef 1/c_\epsilon$ finishes the proof. 
\end{proof}
\comment{To prove that $|h(x_2)-h(x_1)| \le h(x_2-x_1)$, we can do the following. We have $h'(x) = \ln(1+x) - \ln(x)$ which is monotonically decreasing. This shows that $h(\lambda +\gamma) - h(\lambda) = \int_{\lambda}^{\lambda+\gamma} h'(x) dx$ is decreasing. So in particular, $h(x_2) - h(x_1) \leq h(x_2-x_1) - h(0)$.}
Using this claim, we obtain
\begin{align*}
\log \| \fc{f} \|_1 \geq c(1 - 2\alpha_0) \cdot \alpha_0 \log(1/\alpha_0) \cdot \frac{n}{2}  + \frac{1}{2} \log \frac{ (1-\alpha_0)  \left( \alpha_0 (1-2\alpha_0) n^2  - n \right) }{ n^3 (n+1)^2 }.
\end{align*}

This proves the desired result provided $r(f)$ is larger than some constant. Next we handle small (constant) values of $r(f)$. We start with the case $r(f) = 1$. In this case, it is easy to see that $\| \fc{f} \|_1 = O(1)$. Next, we consider $r(f) = 2$. Let $g_k(x) = -1$ iff $|x| = \sum_i x_i = k$. For the function $g_1$, we have for $S \neq \emptyset$,
\begin{align*}
\fc{g_1}(S) &= \frac{1}{2^n} \cdot - 2 \sum_{|x| = 1} \chi_S(x) \\
			&= \frac{-2}{2^n} \cdot  \sum_{i=1}^n (-1)^{\mathbf{1}_{i \in S}} \\
			&= \frac{-2}{2^n} (n - |S| - |S|) = \frac{-2 (n - 2 |S|)}{2^n}.
\end{align*}
Hence,
\begin{align*}
\| \fc{g_1} \|_1 &= 1 - 2\frac{n}{2^n} + \frac{2}{2^n} \sum_{k=1}^n \binom{n}{k} | n - 2k | \\
				&= \Theta(\sqrt{n}),
\end{align*}
by observing that a constant fraction of the probability mass of the binomial distribution lies in the interval $[n/2-2\sqrt{n}, n/2-\sqrt{n}]$. Similarly, one can show that $\| \fc{g_1} + \fc{g_{n-1}} \|_1 = \Theta(\sqrt{n})$. All other functions with $r(f) = 2$ are obtained from these two functions by adding functions $g_0$ or $g_n$ and by multiplying by a constant or the parity function.

We now consider the case $r(f) \geq 3$, but constant. We perform an analysis similar to the proof of Lemma \ref{lem:lb-small-r}. We can assume that $r_0 \geq r_1$. We take $m_0 = \floor{n/2(1 - \sqrt{5 \alpha_0 - 6 \alpha_0^2})}$. \comment{This number was picked so that $5 \alpha_0 - 6 \alpha_0^2$ is divisible by $1-2 \alpha_0$.}
As in \eqref{eq:a0}, we obtain the bound
\[
A_0 \geq \binom{n}{r_0 - 1} 2^{-n} (1-\alpha_0) (2 \alpha_0 n^2 - n).
\]
Hence, the analogue of inequality \eqref{eq:l1-lower-bound} becomes
\begin{align*}
\| \fc{f} \|_1
&\geq \sqrt{ \binom{n}{m_0} \binom{n}{r_0-1} 2^{-n} \frac{ (1-\alpha_0) (2 \alpha_0 n^2 - n) }{n^2} } \\
&\geq \sqrt{ \binom{n}{\floor{n/2(1 - \sqrt{5 \alpha_0 - 6 \alpha_0^2})}
} \binom{n}{\alpha_0 n} 2^{-n} \frac{ (1-\alpha_0) (2 \alpha_0 n^2 - n) }{n^2} }.
\end{align*}
But $\floor{n/2(1 - \sqrt{5 \alpha_0 - 6 \alpha_0^2})} = n/2 - \Theta(\sqrt{n})$ and thus $\binom{n}{\floor{n/2(1 - \sqrt{5 \alpha_0 - 6 \alpha_0^2})} } = \Omega(2^n/\sqrt{n})$ (see inequality \eqref{eq:B3}). As a result,
\begin{align*}
\| \fc{f} \|_1
&\geq \Omega\left( \sqrt{ \frac{1}{\sqrt{n}} \binom{n}{\alpha_0 n}\frac{1}{n}}\right) \\
&\geq \Omega\left( \sqrt{ \binom{n}{r_0 - 1} n^{-3/2} } \right),
\end{align*}
which proves the lemma.
\end{proof}

\subsubsection{Lower Bound: $r \approx n/2$}

For the case $r \approx n/2$, we use a result of \cite{OWZ11} that states that the derivative of a symmetric Boolean function is noise sensitive. Here, we use the noise sensitivity of the derivative $f_{ij}$. The following lemma is an analogue of \cite[Theorem 6]{OWZ11}.
\begin{lemma}
\label{lem:owz}
Let $f$ be a symmetric Boolean function and $f_{ij}$ be defined as in \eqref{eq:second-deriv}. Then for $\rho=1-c/n$, we have
\begin{equation}
\label{eq:g-noise-sensitive}
\sum_S \fc{f_{ij}}(S)^2 \rho^{|S|} \leq \frac{4}{\sqrt{\pi c}} \cdot \sum_S \fc{f_{ij}}(S)^2,
\end{equation}
for any $c \in [1,n]$. Summing over all $i,j$ with $i \neq j$, we get
\begin{equation}\label{eq:1}
8  \sum_S |S|(n-|S|)\fc{f}(S)^2 \rho^{|S|} \leq \frac{4}{\sqrt{\pi c}} \cdot 8 R(f).
\end{equation}
\end{lemma}
\begin{proof}
The proof is the same as the proof of \cite[Theorem 6]{OWZ11} except that we use $f_{ij}$ instead of the derivative. We have
\[
\sum_S \fc{f_{ij}}(S)^2 \rho^{|S|} = \exd{x}{f_{ij}(x) \exd{y}{f_{ij}(y)}}
\]
where $x, y$ are $\rho$-correlated uniform random variables taking values in $\{0,1\}^n$. Note that we can write for any $x$
\begin{align*}
|\exd{y}{f_{12}(y)|x}| &= \left| \exd{y_3 \dots y_n}{\left(\pr{y_1y_2 = 00|x} - \pr{y_1y_2 = 11|x} \right) \left(f(11y_3 \dots y_n) - f(00y_3 \dots y_n) \right)|x} \right| \\
	&\leq \left| \exd{y_3 \dots y_n}{ f(11y_3 \dots y_n) - f(00y_3 \dots y_n)|x} \right|.
\end{align*}
To find an upper bound for this expression, it suffices to replace the use of \cite[Lemma 1]{OWZ11} by the following claim.
\begin{claim}
\label{lem:unimodular}
Let $E = \{i \in [m] : i \equiv 0  \mod 2\}$ and $O = \{i \in [m] : i \equiv 1 \mod 2\}$. Let $p_1, \dots, p_m$ be a non-negative unimodal sequence and $g : [m] \to \{-1,0,1\}$ with the property that the sets $g^{-1}(1) \cap E$ and $g^{-1}(-1) \cap E$ are interleaving, and the sets $g^{-1}(1) \cap O$ and $g^{-1}(-1) \cap O$ are interleaving. Then $|\sum_{i=1}^m p_i g(i)| \leq 2 \max \{p_i\}$.
\end{claim}
To prove the claim, we simply write $|\sum_{i=1}^m p_i g(i)| \leq |\sum_{i \in O} p_i g(i)| + | \sum_{i \in E} p_i g(i)|$. Now \cite[Lemma 1]{OWZ11} implies that  each term is upper-bounded by $\max \{p_i\}$.
\end{proof}

We are now ready to prove the following result.
\begin{lemma}
\label{lem:lb-large-r}
There exists a constant $\gamma < 1/2$ such that for any symmetric Boolean function $f$ with $r(f) \geq \gamma n$, we have
$
\log \| \fc{f} \|_1 = \Omega(n).
$
\end{lemma}
\begin{proof}
Let $\rho = 1-c/n$ where $c$ is a constant chosen later, and let $n$ be large enough so that $\rho \geq 1/2$. We apply \eqref{eq:1} to $g \defeq f \cdot \f{parity}$:
\[
 \sum_S |S|(n-|S|)\fc{g}(S)^2 \rho^{|S|} \leq \frac{4}{\sqrt{\pi c}} \cdot R(g).
\]
Note that $\f{parity}=\chi_{[n]}$ which shows $\fc{f}([n]\setminus S) = \fc{g}(S)$ for all $S$, and in particular $R(g) = R(f)$. So we can rewrite the above inequality as
\begin{equation}\label{eq:2}
 \sum_S |S|(n-|S|)\fc{f}(S)^2 \rho^{n-|S|} \leq \frac{4}{\sqrt{\pi c}} \cdot R(f).
\end{equation}
Summing \eqref{eq:1} and \eqref{eq:2}, we get
\begin{equation}\label{eq:3}
 \sum_S |S|(n-|S|)\fc{f}(S)^2 (1- \rho^{|S|} - \rho^{n-|S|}) \geq \left(1-\frac{8}{\sqrt{\pi c}}\right)R(f).
\end{equation}
Let $\beta < 1/2$ be a positive constant to be chosen later. We have
\begin{align*}
  \sum_{\substack{|S| \leq \beta n }} |S|(n-|S|)\fc{f}(S)^2 (\rho^{|S|} + \rho^{n-|S|})
  & \geq \sum_{\substack{|S| \leq \beta n }} |S|(n-|S|)\fc{f}(S)^2 (\rho^{\beta n} + \rho^{(1-\beta)n}) \\
  & \geq \sum_{\substack{|S| \leq \beta n }} |S|(n-|S|)\fc{f}(S)^2 (1/2 \cdot e^{-c\beta} + 1/2 \cdot e^{-c(1-\beta)}).
\end{align*}
For the first equality, we used the fact that $\rho^{|S|} + \rho^{n-|S|}$ is decreasing in $|S|$ for $|S| \leq n/2$. For the second inequality, we used the inequality $(1-c/n)^{\beta n} \geq e^{-c\beta}/2$ when $1-c/n \geq 1/2$. Similarly, we have
\begin{align*}
  \sum_{\substack{|S| \geq (1-\beta) n }} |S|(n-|S|)\fc{f}(S)^2 (\rho^{|S|} + \rho^{n-|S|})
  & \geq \sum_{\substack{|S| \geq (1-\beta) n }} |S|(n-|S|)\fc{f}(S)^2 (e^{-c\beta}/2 + e^{-c(1-\beta)}/2).
\end{align*}
Summing the two inequalities, we obtain
\begin{align*}
  \sum_{|S| \not\in (\beta n, (1-\beta)n)} |S|(n-|S|)\fc{f}(S)^2 (\rho^{|S|} + \rho^{n-|S|}) 
 & \geq \frac{e^{-c\beta} +e^{-c(1-\beta)}}{2}  \sum_{|S| \not\in (\beta n, (1-\beta)n)} |S|(n-|S|)\fc{f}(S)^2.
\end{align*}
Combining this with (\ref{eq:3}), we obtain
\begin{align*}
& \sum_{\beta n \leq |S| \leq (1-\beta)n} |S|(n-|S|)\fc{f}(S)^2 (1- \rho^{|S|} - \rho^{n-|S|}) \\
&= \sum_{S} |S|(n-|S|)\fc{f}(S)^2 (1- \rho^{|S|} - \rho^{n-|S|}) - \sum_{|S| \not\in (\beta n, (1-\beta)n)} |S|(n-|S|)\fc{f}(S)^2 (1 - \rho^{|S|} - \rho^{n-|S|})\\
&\geq (1-\frac{8}{\sqrt{\pi c}}) R(f) - (1 - e^{-c\beta}/2 - e^{-c(1-\beta)}/2) \sum_{|S| \not\in (\beta n, (1-\beta)n)} |S|(n-|S|)\fc{f}(S)^2.
\end{align*}
As $e^{-c\beta}/2 + e^{-c(1-\beta)}/2 < 1$, this leads to
\[
\sum_{\beta n \leq |S| \leq (1-\beta)n} |S|(n-|S|)\fc{f}(S)^2 (1- \rho^{|S|} - \rho^{n-|S|})
\geq \left( e^{-c\beta}/2 + e^{-c(1-\beta)}/2 - \frac{8}{\sqrt{\pi c}}\right) R(f).
\]
Consequently,
\begin{align*}
 \frac{n^2}{4} \sum_{\beta n \leq |S| \leq (1-\beta)n} \fc{f}(S)^2
 		&\geq R(f) (e^{-c\beta}/2 + e^{-c(1-\beta)}/2 - 8/\sqrt{\pi c}).
\end{align*}
By picking $c = 10^4$ and $\beta = 10^{-4}\ln 2$, we have
$
\frac{e^{-c \beta} + e^{-c(1-\beta)}}{2} - \frac{8}{\sqrt{\pi c}} \geq \frac{1}{10}.
$
We conclude that $\sum_{\beta n \leq k \leq (1-\beta) n} W_k \geq \frac{4R(f)}{10 n^2}$, and thus
\[
\| \fc{f} \|_1 = \sum_{k=0}^n \sqrt{ \binom{n}{k} W_k} \geq \sqrt{ \binom{n}{\beta n} R(f) \frac{4}{10n^2} }.
\]
Using  \eqref{eq:lower-bound-r}, it follows that
\[
\| \fc{f} \|_1 = \Omega \left( \sqrt{\binom{n}{\beta n} \binom{n}{r-1} 2^{-n} } \right)
			= \Omega \left( 2^{(\binent(\beta) + \binent(\alpha) - 1) \frac{n}{2} } (n+1)^{-1}\right),
\]
where $\alpha = (r-1)/n$. If  $\alpha$ is such that $\binent(\alpha) \geq 1 - \binent(\beta)/2$, we obtain the desired bound $\log \| \fc{f} \|_1 = \Omega(n)$.
\end{proof}

\section{Proof of Corollary \ref{cor:par-dec-tree}} \label{sec:cor}
We start by observing that we can assume that $f(x)$ is constant whenever $|x| \in [r_0, n-r_1]$. In fact, if this is not the case, then $f \cdot \f{parity}(x)$ will be constant when $|x| \in [r_0, n-r_1]$. But $f(x)$ can be computed from $f \cdot \f{parity}(x)$ using only one query to $\f{parity}(x)$, which multiplies the size of the tree by at most $2$. In the remainder of the proof, we assume $f(x)$ is constant for $|x| \in [r_0, n-r_1]$.

We start by proving the lower bound. It is simple to prove that $\| \fc{f} \|_1$ is a lower bound on the parity decision tree size of $f$ \cite[Lemma 5.1]{KM91}. For completeness, we provide a sketch of a proof. As all the possible inputs that lead to some leaf $L$ have the same value for $f$, we can write $f$ as a sum over all leaves of the tree $f(x) = \sum_{L} f(L) \mathbf{1}_{L}(x)$, where the function $\mathbf{1}_L$ takes value $1$ if the input belongs to the leaf $L$ and is $0$ otherwise. By linearity of the Fourier transform and the triangle inequality, we have
$\| \fc{f} \|_1 \leq \sum_L |f(L)| \| \fc{\mathbf{1}_{L}} \|_1$. Now observe that the inputs corresponding to $L$ (that we also call $L$) are inputs that satisfy some parity conditions on subsets belonging to some subspace $\cS$. Then, we have $\fc{\mathbf{1}_{L}}(S) = \pm \frac{|L|}{2^n}$ for any $S \in \cS$. Note that the number of such subsets is $2^n/|L|$. But if $S \notin \cS$, then $\sum_{x \in L} \chi_S(x) = 0$. It follows that $\| \fc{\mathbf{1}_L} \|_1 = 1$ and that $\| \fc{f} \|_1$ is a lower bound on the size of the tree.

Using Theorem \ref{thm:l1-sym}, this proves the lower bound stated in Corollary \ref{cor:par-dec-tree}, except in the case where $r(f) = 1$. For this case, observe that we can assume that a leaf at depth $d$ corresponds to $2^{n-d}$ possible inputs; see e.g., \cite[Lemma 5.1]{KM91}. So we have at most two input bit strings that have a value for $f$ that is different from the value $f$ takes when $x \in [r_0(f), n-r_1(f)]$. This proves that the depth of the tree is at least $n-1$ and completes the proof of the lower bound.

For the upper bound, we give a decision tree of size at most $4 \binom{n}{r_0(f)} + 4 \binom{n}{r_1(f)}$ for computing $f$. We start by considering a complete binary tree of depth $n$. Level $i$ of the tree corresponds to querying the $i$-th input bit $x_i$. The number of leaves of the tree is $2^{n}$. Clearly, one can compute any function using such a tree. We are going to use the values $r_0(f)$ and $r_1(f)$ to remove unnecessary nodes from the tree. Note that each node at level $i$ can be labeled by a bit string of length $i$. We remove all the nodes that have $r_0$ ones and at least $r_1$ zeros, and the nodes that have $r_1$ zeros and at least $r_0$ ones, together with all their children. All of these nodes correspond to inputs $x$ for which $|x| \in [r_0, n-r_1]$, so the value of $f$ is a constant that only depends on $f$.

It now remains to compute the number of leaves of the constructed decision tree. The number of leaves at a level $i < n$ is $0$ if $i < r_0 + r_1$ and $\binom{i}{r_0-1} + \binom{i}{r_1-1}$ if $i \geq r_0 + r_1$. At level $n$, we have all the remaining nodes that can have at most $r_0$ ones or at most $r_1$ zeros, thus at most $\binom{n}{r_0} + \binom{n}{r_1}$ leaves. Thus, the total number of leaves is at most
\begin{align*}
\sum^{n-1}_{i=r_0+r_1} \binom{i}{r_0-1} + \binom{i}{r_1-1} + \binom{n}{r_0} + \binom{n}{r_1}
&\leq \sum_{i=r_0-1}^{n-1} \binom{i}{r_0-1} + \sum_{i=r_1-1}^{n-1} \binom{i}{r_1-1} + \binom{n}{r_0} + \binom{n}{r_1} \\
&= 2 \cdot \left( \binom{n}{r_0} + \binom{n}{r_1} \right).
\end{align*}
We can then obtain the stated result by \eqref{eq:B2-3} and the fact that $h(x) \le -2 x \log x$ for $x \in (0,1/2]$.

\section{Conclusion and Future Work}\label{sec:futurework}
A natural next step is to extend Theorem \ref{thm:l1-sym} to \emph{approximate} spectral norm. Indeed this would have  interesting implications. Recall that the $\eps$-approximate spectral norm of a Boolean function $f$ is the smallest spectral norm of a function $g$ with $\norm{f-g}_\infty \leq \eps$, i.e., for all $x$, $|f(x) - g(x)| \leq \eps$. Trivially $\norm{\fc{f}}_{1,\eps}$ is smaller than $\norm{\fc{f}}_1$. We conjecture that it cannot be much smaller.

\begin{conjecture}\label{conj:approximateell1}
For all symmetric functions $f:\{0,1\}^n \to \{\pm 1\}$,
\[
 \log \norm{\fc{f}}_1 = \Theta^*(\log \norm{\fc{f}}_{1,1/3})
\]
where $\Theta^*$ suppresses $O(\log n)$ factors.
\end{conjecture}

We now discuss some of the applications of the above conjecture in conjunction with Theorem \ref{thm:l1-sym}.

\subsubsection*{Analog of Paturi's Result for Monomial Complexity}

A famous result of Paturi \cite{Pat92} characterizes the approximate degree of all symmetric functions. Recall that the degree of a function $f$ is the largest $|S|$ such that $\fc{f}(S)$ is non-zero. Let $t_0$ and $t_1$ be the minimum integers such that $f(i) = f(i+1)$ for all $i \in [t_0,n-t_1]$.
\begin{theorem}[\cite{Pat92}]
 Let $f:\{0,1\}^n \to \{\pm 1\}$ be a symmetric function and let $t_0$ and $t_1$ be defined as above. Then,
$\deg_{1/3}(f) = \Theta(\sqrt{n(t_0 + t_1)})$.
\end{theorem}
Paturi's result has found numerous applications in theoretical computer science \cite{Raz03,BBCMW01,She09,Wol08, She11}. 

The monomial complexity of a Boolean function $f$, denoted $\mon(f)$, is the number of non-zero Fourier coefficients of $f$. The monomial complexity appears naturally in various areas of complexity theory, and it is desirable to obtain simple characterizations for natural classes of functions. 
An argument similar to the one in~\cite{BS92} shows that $\mon_{\epsilon}(f)  \le \frac{2n}{\epsilon^2} \|\fc{f}\|_1^2$ for every $\epsilon>0$. Combining this with Conjecture \ref{conj:approximateell1} and Theorem \ref{thm:l1-sym} would show that $r(f)$ characterizes the approximate monomial complexity of $f$:
\begin{conjecture}[Consequence of Conjecture~\ref{conj:approximateell1}]
For a symmetric function $f:\{0,1\}^n \to \{\pm 1\}$,
$$
 \log \mon_{1/3}(f) = \Theta^*(r(f)).
$$
\end{conjecture}

\subsubsection*{Communication Complexity of Xor Functions}

Recall the Log Rank Conjecture mentioned in the introduction. This conjecture has an analogous version for the randomized communication complexity model: ``Log Approximation Rank Conjecture''. The $\eps$-approximate rank of a matrix $M$ is denoted by $\rank_\eps (M)$, and is the minimum rank of a matrix that $\eps$ approximates $M$. Denote by $\bfR^\eps(F)$ the $\eps$-error randomized communication complexity of $F$. It is known that $\bfR^\eps(F) \geq \log \rank_{\eps'}(M_F)$, where $\eps'$ is a constant that depends on $\eps$ and $M_F$ is the matrix representation of $F$. Log Approximation Rank Conjecture states that this lower bound is tight:

\begin{conjecture}[Log Approximation Rank Conjecture]
 There is a universal constant $c$ such that for any 2 party communication problem $F$,
$$
 \log \rank_{\eps'} (M_F) \leq \bfR^\eps(F) \leq \log^c  \rank_{\eps'} (M_F).
$$
\end{conjecture}

The important paper of Razborov \cite{Raz03} established this conjecture for the functions $F(x,y) = f(x \wedge y)$ where $f$ is symmetric. In fact, Razborov showed that the quantum and classical randomized communication complexities of such functions are polynomially related. Later, Shi and Zhang \cite{SZ09a}, via a reduction to the case $f(x \wedge y)$, showed the quantum/classical equivalence for symmetric xor functions $F(x,y) = f(x \xor y)$. They show that the randomized and quantum bounded error communication complexities of $F$ are both $\Theta(r(f))$, up to polylog factors. However, their result does not verify the Log Approximation Rank Conjecture for symmetric xor functions.

Conjecture \ref{conj:approximateell1} along with Theorem~\ref{thm:l1-sym} would verify the Log Approximation Rank Conjecture for symmetric xor functions (This follows from the protocol of Shi and Zhang \cite[Proposition 3.4]{SZ09a} for symmetric xor functions, and the facts $\norm{M_F}_{tr,\eps} = 2^n\norm{\fc{f}}_{1,\eps}$ and $\rank_\eps(M_F)^{1/2} \geq \norm{M_F}_{tr,\eps}/(1+\eps)2^n$.). Furthermore, we would obtain a direct proof of the result of Shi and Zhang. This is very desirable since a major open problem is to understand the communication complexity of $f(x \xor y)$ for general $f$ (with no symmetry condition on $f$). There is a sentiment that this should be easier to tackle than $f(x \wedge y)$ as xor functions seem more amenable to Fourier analytic techniques. A direct proof of the result of Shi and Zhang gives more insight into the communication complexity of xor functions.

\subsubsection*{Agnostically Learning Symmetric Functions}

Let $\calC$ be a concept class and $g_i : \{-1,1\}^n \to \bbR$ be functions for $1 \leq i \leq s$ such that every $f: \{-1,1\}^n \to \{-1,1\}$ in $\calC$ satisfies
$
 \norm{f- \sum_{i=1}^s c_i g_i}_\infty \leq \eps,
$
for some reals $c_i$. The smallest $s$ for  which such $g_i$'s exist  corresponds to the $\eps$-approximate rank of $\calC$. If each $g_i(x)$ is computable in polynomial time, then $\calC$ can be agnostically learned under any distribution in time $\poly(n,s)$ and with accuracy $\eps$ \cite{KKMS08}.

Klivans and Sherstov \cite{KS10} proved strong lower bounds on the approximate rank of the concept class of disjunctions $\{\bigvee_{i \in S} x_i : S \subseteq [n]\}$ and majority functions $\{\f{maj}(\pm x_1, \pm x_2, \ldots, \pm x_n)\}$ thereby ruled out the possibility of applying the algorithm of \cite{KKMS08} to agnostically learning these concept classes.

Theorem~\ref{thm:l1-sym} together with Conjecture \ref{conj:approximateell1} provides additional negative results and gives strong lower bounds on the approximate rank of the concept class consisting of symmetric functions $f$ with large $r(f)$.

%
%

\bibliographystyle{alpha}
\bibliography{boolean-functions}

\appendix


\end{document}